\documentclass[12pt]{article}
\usepackage[utf8]{inputenc}
\usepackage{amsthm,amsmath,amssymb,bbm,hyperref,graphicx,float}
\usepackage[dvipsnames]{xcolor}
\allowdisplaybreaks

\title{Time Evolution of Typical Pure States from a Macroscopic Hilbert Subspace}
\author{
Stefan Teufel\footnote{Mathematics Institute, Eberhard Karls University T\"ubingen, 
	Auf der Morgenstelle 10, 72076 T\"ubingen, Germany. ORCID: 0000-0003-3296-4261,
	E-mail: stefan.teufel@uni-tuebingen.de},~
Roderich Tumulka\footnote{Mathematics Institute, Eberhard Karls University T\"ubingen, 
	Auf der Morgenstelle 10, 72076 T\"ubingen, Germany. 
	ORCID: 0000-0001-5075-9929
	E-mail: roderich.tumulka@uni-tuebingen.de},~ 
Cornelia Vogel\footnote{Mathematics Institute, Eberhard Karls University T\"ubingen, 
	Auf der Morgenstelle 10, 72076 T\"ubingen, Germany.
	ORCID: 0000-0002-3905-4730,
	E-mail: cornelia.vogel@uni-tuebingen.de}
}
\date{December 19, 2022}

\newtheorem{thm}{Theorem}
\newtheorem{prop}{Proposition}
\newtheorem{lemma}{Lemma}
\newtheorem{cor}{Corollary}
\theoremstyle{definition}

\theoremstyle{remark}

\DeclareMathOperator{\tr}{tr}
\DeclareMathOperator{\Var}{Var}

\newcommand{\IM}{\textup{Im\,}}
\newcommand{\Hilbert}{\mathcal{H}}
\newcommand{\EEE}{\mathbb{E}}

\newcommand{\RRR}{\mathbb{R}}
\newcommand{\SSS}{\mathbb{S}}
\newcommand{\be}{\begin{equation}}
\newcommand{\ee}{\end{equation}}

\newcommand{\scp}[2]{\langle #1|#2 \rangle}
\newcommand{\prj}{\Pi}

\begin{document}
\maketitle

\begin{abstract}
We consider a macroscopic quantum system with unitarily evolving pure state $\psi_t\in \Hilbert$ and take it for granted that different macro states correspond to mutually orthogonal, high-dimensional subspaces $\Hilbert_\nu$ (macro spaces) of $\Hilbert$. Let $P_\nu$ denote the projection to $\Hilbert_\nu$. We prove two facts about the evolution of the superposition weights $\|P_\nu\psi_t\|^2$: First, given any $T>0$, for most initial states $\psi_0$ from any particular macro space $\Hilbert_\mu$ (possibly far from thermal equilibrium), the curve $t\mapsto \|P_\nu \psi_t\|^2$ is approximately the same (i.e., nearly independent of $\psi_0$) on the time interval $[0,T]$. And second, for most $\psi_0$ from $\Hilbert_\mu$ and most $t\in[0,\infty)$, $\|P_\nu \psi_t\|^2$ is close to a value $M_{\mu\nu}$ that is independent of both $t$ and $\psi_0$. 
The first is an instance of the phenomenon of dynamical typicality observed by Bartsch, Gemmer, and Reimann, and the second modifies, extends, and in a way simplifies the concept, introduced by von Neumann, now known as normal typicality.

\medskip

\end{abstract}

\section{Introduction}

The approach of studying thermalization through the analysis of closed quantum systems with huge numbers of degrees of freedom has led, among other things, to the \emph{eigenstate thermalization hypothesis} (ETH) \cite{Deutsch91,Srednicki94,GogEis16}, to the discovery of \emph{canonical typicality} \cite{GMM04,PSW06,GLTZ06}, and more recently to the discovery of \emph{dynamical typicality} \cite{BRGSR18,BG09,MGE,Reimann2018b,Reimann2018a,RG20}, which is the fact that most pure states $\psi$ with a given quantum expectation value $\scp{\psi}{A|\psi}$ of a macroscopic observable $A$ also have nearly the same $\scp{\psi}{B|\psi}$ for any other observable $B$ (and likewise also nearly the same $\scp{\psi_t}{B|\psi_t}$). 
Here, we provide a very simple proof of an important special case of this statement, namely for $A$ a projection and $\scp{\psi}{A|\psi}=1$. Put differently, we show that  most $\psi$ from a macroscopically large subspace of Hilbert space have almost the same expectation values of bounded observables.

Our second result concerns the long-time behavior of $\scp{\psi_t}{B|\psi_t}$ under the unitary evolution $\psi_t=\exp(-iHt)\psi_0$ (taking $\hbar=1$) and extends previous results of Reimann and Gemmer \cite{RG20} as well as von Neumann's \cite{vonNeumann29} result now known as \emph{normal typicality} \cite{GLMTZ09,GLTZ10}. In particular, our result avoids certain unrealistic assumptions of von Neumann's.

As usual for the description of  macroscopic closed quantum systems, we restrict our consideration to a micro-canonical energy interval $[E-\Delta E,E]$ that is small in macroscopic units but large enough to contain very many eigenvalues of the Hamiltonian~$H$; for a system of $N$ particles, relevant intervals contain of order $\exp(N)$ eigenvalues. Let~$\Hilbert$ be the corresponding spectral subspace, i.e., the range of $\mathbbm{1}_{[E-\Delta E,E]}(H)$, or \textit{energy shell}, and let $\mathbb{S}(\Hilbert) = \{\psi \in \Hilbert: \|\psi\|=1\}$ denote the unit sphere and $D:=\dim\Hilbert<\infty$. Following von Neumann \cite{vonNeumann29}, we assume that different macro states $\nu$ of the system correspond to mutually orthogonal subspaces $\Hilbert_\nu$ (macro spaces) of $\Hilbert$ such that
\begin{align}
    \Hilbert = \bigoplus_{\nu}\Hilbert_\nu.\label{eq: decomposition}
\end{align}
Different vectors in the same $\Hilbert_\nu$ are regarded as ``looking macroscopically equal''. For example, the ``macroscopic look'' could be defined in terms of mutually commuting self-adjoint operators $M_1,\ldots,M_K$ regarded as the ``macroscopic observables'' \cite{vonNeumann29}; then $\Hilbert_\nu$ are the joint eigenspaces and $\nu=(m_1,\ldots,m_K)$ is the corresponding list of eigenvalues.
Let $P_\nu$ denote the projection onto $\Hilbert_\nu$. Although some macro spaces will have much larger dimensions $d_\nu:=\dim\Hilbert_\nu$ than others, all $d_\nu$ will be very large, roughly comparable to $\exp(N)$. 

In this setting, it is natural to consider initial states $\psi_0$ from a certain macro space and ask about the time evolution of the \textit{macroscopic superposition weights} $\|P_\nu \psi_t\|^2$. 
We present two general, theoretical findings about these weights that mainly arise just from the hugeness of the $d_\nu$'s. The first finding (dynamical typicality) is that the curve given by $\|P_\nu\psi_t\|^2$ as a function of $t$ is nearly $\psi_0$-independent once we fix the macro state of $\psi_0$. In other words, if $\psi_0$ is purely random in $\Hilbert_\mu$, then the superposition weights are nearly deterministic. 
The second finding (generalized normal typicality) is that in the long run, as $t\to\infty$, $\|P_\nu\psi_t\|^2$ is nearly constant, meaning it is close for most $t\in[0,\infty)$ to a $t$-independent and $\psi_0$-independent value, once we fix the macro state of $\psi_0$. 
This does not mean that $\|P_\nu\psi_t\|^2$ converges as $t\to\infty$ (it does not), but that the time periods in which $\|P_\nu\psi_t\|^2$ is far from that value tend to be short compared to the time intervals separating these periods. One can say that the $\|P_\nu\psi_t\|^2$ \emph{equilibrate} in the long run; however, this equilibration does not correspond to thermal equilibrium in the sense of thermodynamics; rather, thermal equilibrium at time $t$ would correspond to $\|P_\nu\psi_t\|^2\approx 1$ for one particular $\nu$ (the macro state of thermal equilibrium, $\Hilbert_\nu=\Hilbert_{\textup{eq}}$) and $\|P_\nu\psi_t\|^2\approx 0$ for all other $\nu$'s. We therefore speak of \emph{normal equilibrium} when $\|P_\nu\psi_t\|^2$ assumes its long-term value for all $\nu$.

Our results are \emph{typicality} statements, i.e., they concern the way \emph{most} $\psi_0$ behave, notwithstanding the existence of few exceptional $\psi_0$ that behave differently. However, a statement about most $\psi_0$ in $\SSS(\Hilbert)$ would be of limited interest because it could be violated by every system outside of thermal equilibrium, as usually most $\psi_0$ in $\SSS(\Hilbert)$ are in thermal equilibrium (meaning they are close to $\Hilbert_{\textup{eq}}$) \cite{GLMTZ10}. Instead, we make more specific statements: we allow an arbitrary initial macro space $\Hilbert_\mu$, possibly far from thermal equilibrium, and make statements about most $\psi_0$ in $\SSS(\Hilbert_\mu)$. Such statements are also naturally of interest when we ask about the increase of the quantum Boltzmann entropy observable \cite{GLTZ20}
\be
\hat S = \sum_\nu S(\nu) P_\nu\,,
\ee
where 
\be
S(\nu)=k_B \log d_\nu
\ee
is the quantum Boltzmann entropy of the macro state $\nu$, and $k_B$ is the Boltzmann constant. Note that a quantum system can be in a superposition of different macro states and thus also in a superposition of different entropy values.

In Section~\ref{sec:dyntyp}, we formulate 
our theorem about 
dynamical typicality and compare it to related results in the literature. In Section~\ref{sec:normtyp}, the same for generalized normal typicality. In Section~\ref{sec:proofdyntyp}, we prove our result on dynamical typicality. In Section~\ref{sec:proofstrategies}, we formulate further variants of our results. In Section~\ref{sec:realistic}, conclusions for realistic sizes of $d_\nu$ are discussed. In Section~\ref{sec:outlinepfGNT}, we outline the proof of generalized normal typicality. In Section~\ref{sec:proofs}, we collect the remaining proofs. In Section~\ref{sec:conclusions}, we conclude.

\section{Dynamical Typicality}
\label{sec:dyntyp}

\subsection{Mathematical Description}

For formulating theorems,
we introduce the following terminology. Suppose that for each $\psi\in\SSS(\Hilbert_\mu)$, the statement $s(\psi)$ is either true or false, and let $\varepsilon>0$. We say that $s(\psi)$ is true for $(1-\varepsilon)$-\textit{most} $\psi\in\SSS(\Hilbert_\mu)$ if and only if
\begin{align}
    u_\mu\big(\big\{\psi\in\mathbb{S}(\Hilbert_\mu): s(\psi)\big\}\big) \geq 1-\varepsilon\,,
\end{align}
where $u_\mu$ is the normalized uniform measure over $\SSS(\Hilbert_\mu)$. Similarly, given $T>0$ and $\delta>0$, we say that a statement $s(t)$ is true for $(1-\delta)$-\textit{most} $t\in [0,T]$ if and only if
\be\label{mosttdef1}
  \tfrac{1}{T} \big|\big\{t\in [0,T]: s(t)\big\}\big| \geq 1-\delta\,,
\ee
where $|S|$ means the length of the set $S\subset \mathbb{R}$; and that $s(t)$ is true for $(1-\delta)$-\textit{most} $t\in [0,\infty)$ if and only if the lim inf of the left-hand side of \eqref{mosttdef1} as $T\to\infty$ is $\geq 1-\delta$.

The first finding we mentioned
can be expressed as follows. 

\begin{thm}[Dynamical typicality]\label{thm:dyntyp}
Let $\mu, \nu$ be arbitrary macro states. There is a function $w_{\mu\nu}:\mathbb{R} \to [0,1]$ such that for every $t\in\mathbb{R}$ and every $\varepsilon>0$, for $(1-\varepsilon)$-most $\psi_0 \in \mathbb{S}(\Hilbert_\mu)$,
\begin{equation}\label{dyntyp1}
\Bigl| \|P_\nu \psi_t\|^2-w_{\mu\nu}(t)\Bigr| \leq \frac{1}{\sqrt{\varepsilon d_\mu}} \,.
\end{equation}
Moreover, for every $\mu,\nu$, every $T>0$, and $(1-\varepsilon)$-most $\psi_0\in\SSS(\Hilbert_\mu)$,
\begin{equation}\label{dyntyp2}
\frac{1}{T}\int_0^T \!  \bigl| \|P_\nu \psi_t\|^2-w_{\mu\nu}(t)\bigr|^2  dt \leq \frac{1}{\varepsilon d_\mu} \,.
\end{equation}
\end{thm}

That is, if $d_\mu\gg 1/\varepsilon$, then for any $t$ and purely random $\psi_0$ from $\Hilbert_\mu$, the random value $\|P_\nu\psi_t\|^2$ is very probably close to the non-random value $w_{\mu\nu}(t)$. The latter can in fact be taken to be the average of $\|P_\nu\psi_t\|^2$ over $\psi_0\in\SSS(\Hilbert_\mu)$, which is
\be
w_{\mu\nu}(t) := \frac{1}{d_\mu} \tr\Bigl[ P_\mu \exp(iHt) P_\nu \exp(-iHt)\Bigr] \,.
\ee
Likewise, the whole curve of $\|P_\nu\psi_t\|^2$ as a function of $t\in[0,T]$ is very probably close, in the $L^2$ norm, to $w_{\mu\nu}(t)$ as a function of $t$. (Smallness of the $L^2$ norm implies further that $\bigl| \|P_\nu \psi_t\|^2-w_{\mu\nu}(t)\bigr|$ is small for most $t$; however, this statement, which is equivalent to saying that the expression is small for most pairs $(t,\psi_0)\in [0,T]\times \SSS(\Hilbert_\mu)$, follows already from \eqref{dyntyp1}; note that the quantifiers ``most $t$'' and ``most $\psi_0$'' commute. Moreover, it also follows from \eqref{dyntyp2} by letting $T\to\infty$ that the long-time average of $\bigl| \|P_\nu \psi_t\|^2-w_{\mu\nu}(t)\bigr|^2$ is small, but this statement is actually weaker than for finite $T$, and it will be superseded below by a more specific statement in our second result, generalized normal typicality.)
A more general statement for arbitrary operators $B$ instead of $P_\nu$ and a tighter error bound is formulated in Section~\ref{sec:proofstrategies}.

As a further remark, we observe that another quantity is also deterministic for purely random $\psi_0$ from $\SSS(\Hilbert_\mu)$: not only is the probability $\|P_\nu \psi_t\|^2$ associated with $\Hilbert_\nu$ at time $t$ nearly deterministic, but also the \emph{probability current} between $\Hilbert_\nu$ and $\Hilbert_{\nu'}$,
\be\label{Jdef}
J_{\nu\nu'} := -i\left(\langle\psi_t|P_\nu H P_{\nu'}|\psi_t\rangle -\langle\psi_t|P_{\nu'}HP_\nu|\psi_t\rangle\right) =  2\, \IM \scp{\psi_t}{P_\nu H P_{\nu'}|\psi_t} \,.
\ee
This quantity expresses the amount of probability passing, per unit time, from $\nu'$ to $\nu$ minus that from $\nu$ to $\nu'$; it satisfies a discrete version of the continuity equation, viz.,
\begin{align}
    \partial_t \|P_\nu\psi_t\|^2 = \sum_{\nu'} J_{\nu\nu'}.
\end{align}
In Section~\ref{sec:current} we will show that the probability current between two macro spaces is deterministic.

\subsection{Previous Results about Dynamical Typicality}

Bartsch and Gemmer \cite{BG09} introduced the name ``dynamical typicality'' for the following closely related phenomenon: Given an observable $A$ and $a\in\mathbb{R}$, there is a function $a(t)$ such that for every $t\in\mathbb{R}$ and most $\psi_0\in\SSS(\Hilbert)$ with $\scp{\psi_0}{A|\psi_0} \approx a$, $\scp{\psi_t}{A|\psi_t} \approx a(t)$. M\"uller, Gross, and Eisert \cite{MGE} proved a rigorous version of this fact that also implies  that for every operator $B$ whose operator norm (largest absolute eigenvalue or singular value) is not too large, there is a value $b$ such that for most $\psi_0\in\SSS(\Hilbert)$ with $\scp{\psi_0}{A|\psi_0} \approx a$, $\scp{\psi_0}{B|\psi_0}\approx b$.  As Reimann \cite{Reimann2018a} pointed out, this also implies that for every $t\in\mathbb{R}$ and most $\psi_0\in\SSS(\Hilbert)$ with $\scp{\psi_0}{A|\psi_0} \approx a$, $\scp{\psi_t}{B|\psi_t} \approx b(t)$ for suitable $b(t)$. Setting $A=P_\mu$, $a=1$, and $B=P_\nu$, this yields that for every $t\in\mathbb{R}$ and most $\psi_0\in\SSS(\Hilbert_\mu)$, $\scp{\psi_t}{P_\nu|\psi_t}=\|P_\nu\psi_t\|^2$ is nearly deterministic. For technical reasons, the proofs of M\"uller, Gross, and Eisert \cite{MGE} and Reimann \cite{Reimann2018a} do not actually cover the case that $A$ is a projection and $a=1$. As was pointed out to us by one of the referees of our paper,    Balz et al.\ \cite{BRGSR18} provide a general result that covers Theorem~\ref{thm:dyntyp} as a special case.  Although our proof strategy is similar to the one in \cite{BRGSR18}, we decided to present our proof in this paper, because it is very simple and transparent and could help to make the at first sight striking phenomenon of dynamical typicality a text book result. 
Theorem~\ref{thm:dyntyp} can also be obtained through a proof strategy used by Reimann and Gemmer \cite{RG20}.

A further related result is given by Strasberg et al.\ \cite{SWGW22}, who consider repeated measurements at $0<t_1<t_2<\ldots<t_r<T$ of all $P_\nu$'s and argue that the probability distribution of the outcomes is essentially indistinguishable from the joint distribution of $X_{t_1},\ldots,X_{t_r}$ for a suitable Markov process $X_t$ on the set of $\nu$'s. This includes the claim that omitting one of the measurements does not significantly alter the distribution of the other outcomes, so the distribution of $X_t$ should agree with $\|P_\nu \psi_t\|^2$, which is in line with our result.

\section{Generalized Normal Typicality}
\label{sec:normtyp}

\subsection{Motivation}

It is well known that for most $\phi \in \mathbb{S}(\Hilbert)$, 
\begin{align}
    \|P_\nu\phi\|^2 \approx \frac{d_\nu}{D}\,, \label{eq: Pnuphi}
\end{align}
provided that $d_\nu$ and $D:=\dim\Hilbert$ are large~\cite{GLMTZ09}. 
Under the additional condition that relative to a fixed decomposition \eqref{eq: decomposition} into macro spaces  the eigenbasis of $H$ 
is chosen purely randomly among all orthonormal bases (and some further technical conditions that are not very restrictive),  
\eqref{eq: Pnuphi} holds also for the eigenstates of $H$, and it can be shown that
every $\psi_0\in\mathbb{S}(\Hilbert)$ evolves so that for most times $t$,
\begin{align}
    \|P_\nu\psi_t\|^2 \approx \frac{d_\nu}{D}\,.
\end{align}
This fact is known as \textit{normal typicality} \cite{vonNeumann29,GLMTZ09,GLTZ10,Reimann2015}. 

The assumption of a purely random eigenbasis can be regarded as expressing that the energy eigenbasis is unrelated to the orthogonal decomposition \eqref{eq: decomposition}.
In most realistic systems, however, the energy eigenbasis and the macro decomposition \eqref{eq: decomposition} are  not unrelated. If they were unrelated, then the system would very rapidly go from any macro space $\Hilbert_\nu$ directly to the thermal equilibrium macro space $\Hilbert_{\textup{eq}}$ (a macro space containing most dimensions of $\Hilbert$, $d_{\textup{eq}}/D\approx 1$) \cite{GHT13,GHT15,GHT14}.
But that does not happen in most systems because thermal equilibrium requires that energy (and other quantities) is rather evenly distributed over all degrees of freedom, and for getting evenly distributed, it needs to get transported through space, which usually requires time and passage through other macro states, cf.\ Figure~\ref{fig1}.

\begin{figure}[h]
\centering
 \includegraphics[width=0.7\textwidth]{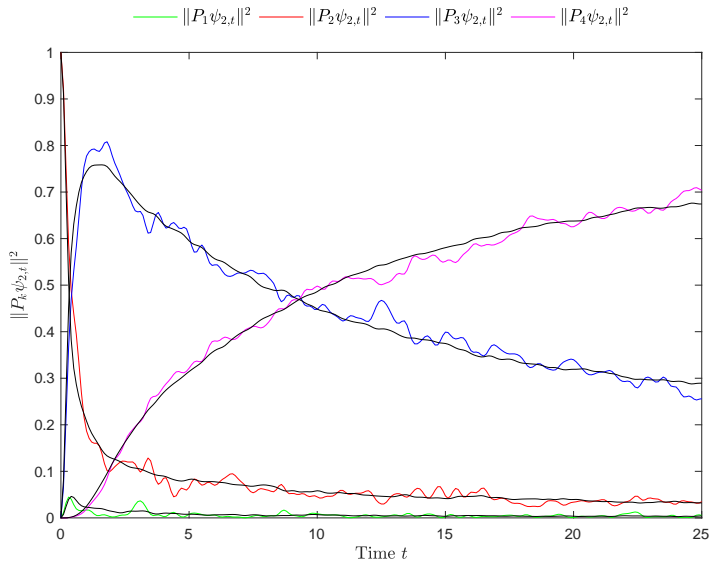}
 \caption{\label{fig1} 
 Example of time evolution of superposition weights $\|P_\nu \psi_t\|^2$, here in a Hilbert space of dimension $D=2222$ decomposed into 4 macro spaces of dimensions $d_1=2$ (green curve), $d_2=20$ (red curve), $d_3=200$ (blue curve), and $d_4=2000$ (purple curve). The four curves add up to 1 at each $t$. At large $t$, the equilibrium subspace $\Hilbert_4$ has the biggest contribution. $\psi_0$ was chosen purely randomly from $\SSS(\Hilbert_2)$ (i.e., $\mu=2$, so the red curve starts at 1, all others at 0). 
 The Hamiltonian is a random band matrix (i.e., only entries sufficiently close to the main diagonal are significantly nonzero) in a basis aligned with the macro spaces,  but with a wide enough bandwidth to still ensure delocalized eigenfunctions.
 Thus, parts of $\psi_t$ reach $\Hilbert_4$ only after passing through $\Hilbert_3$, as mirrored in the fact that the blue curve increases first before it decreases in favor of the purple curve. Along with each of the four curves, also its deterministic approximation $w_{2\nu}(t)$ (in black) is drawn; dynamical typicality asserts that it is a good approximation. 
}
\end{figure}

That is why we are interested in generalizations of normal typicality that apply also to Hamiltonians whose eigenbasis is not unrelated to $\Hilbert_\nu$. For such $H$, eigenvectors $\phi$ must be expected to have superposition weights $\|P_\nu\phi\|^2$ not always near $d_\nu/D$. Our result actually applies to \textit{all} Hamiltonians, at the expense that it does not apply to \textit{all} initial quantum states $\psi_0$. As noted already, a statement about \textit{most} $\psi_0\in\mathbb{S}(\Hilbert)$ would be limited to systems starting out in thermal equilibrium. 
Our result states that for any macro state $\mu$, most $\psi_0\in\SSS(\Hilbert_\mu)$  evolve so that for most times $t$
\begin{align}
  \|P_\nu\psi_t\|^2 \approx M_{\mu\nu}\,,
  \end{align}
provided that $d_\mu$ is large. See Theorem~\ref{thm: GNT Pnu} for the precise quantitative statement and the definition of $M_{\mu\nu}$. The proof (see Section~\ref{sec:proofs}) builds particularly on techniques developed by Short and Farrelly \cite{Short11,SF12}, but is also related to a series of works on quantum equilibration (e.g., \cite{Reimann08,LPSW09}) in which the long-time behavior of $\scp{\psi_t}{B|\psi_t}$ is studied under various assumptions on $B$ and $\psi_0$.

The $M_{\mu\nu}$ are actually the averages of $\|P_\nu \psi_t\|^2$ over $t\in[0,\infty)$ and over $\psi_0\in \SSS(\Hilbert_\mu)$. Thus, they
depend only on $H$ and the decomposition \eqref{eq: decomposition}, but not on $t$ or $\psi_0$. 

In this setting, \emph{thermalization} means that $M_{\mu\,\mathrm{eq}}\approx 1$ for every $\mu$, i.e., that for all macro states $\mu$ the overwhelming majority of micro states eventually reach thermal equilibrium in the sense that $\psi_t$ lies almost completely in $\Hilbert_\mathrm{eq}$ and spends most of the time in the long run there.
The time scale on which thermalization happens can be read off from the function $w_{\mu\,\mathrm{eq}}(t)$, while the other $w_{\mu\nu}(t)$ provide   information about the detailed path to thermal equilibrium passing  through intermediate macro states.

\subsection{Statement of Result}

In the following we consider Hamiltonians with spectral decomposition
\begin{align}
H=\sum_{e\in\mathcal{E}} e \,\prj_e,
\end{align}
where $\mathcal{E}$ is the set of  distinct eigenvalues of $H$ and $\prj_e$ the projection onto the eigenspace of $H$ with eigenvalue $e$. The  quantitative bounds in our theorem depend on the Hamiltonian only through the following characteristics of the distribution of its eigenvalues: the maximum degeneracy   $D_E := \max_{e\in\mathcal{E}} \tr(\prj_e)$ of an eigenvalue and the maximal gap degeneracy
\be\label{DGdef}
D_G := \max_{E\in\RRR} \# \bigl\{ (e,e')\in\mathcal{E}\times\mathcal{E}: e\neq e' \text{ and }e-e'=E \bigr\}\,.
\ee

\begin{thm}[Generalized normal typicality]\label{thm: GNT Pnu}
 Let $\mu,\nu$ be any macro states and define 
\begin{align}
    M_{\mu \nu} &:= \frac{1}{d_\mu} \sum_{e\in\mathcal{E}} \tr\left(P_\mu \prj_e P_\nu \prj_e\right)\,.
\end{align}
Then for any  $\varepsilon, \delta>0$,
 $(1-\varepsilon)$-most $\psi_0\in\mathbb{S}(\Hilbert_\mu)$ are such that for $(1-\delta)$-most $t\in[0,\infty)$ 
\begin{align}
    \biggl| \|P_\nu\psi_t\|^2 - M_{\mu\nu} \biggr|\leq
    4\,\sqrt{\frac{D_E D_G}{\delta\varepsilon d_\mu}\min\left\{1,\frac{d_\nu}{d_\mu}\right\}}\,.\label{ineq: GNT Pnu}
\end{align}
\end{thm}

Thus, as soon as $d_\mu \gg D_E D_G$, i.e., as soon as the dimension of $\Hilbert_\mu$ is huge and no eigenvalue and no gap of $H$ is macroscopically degenerate, for most initial states $\psi_0\in\mathbb{S}(\Hilbert_\mu)$ the superposition weight $\|P_\nu\psi_t\|^2$ will be close to the fixed value $M_{\mu\nu}$ for most times $t$.

\begin{figure}[h]
\centering
  \includegraphics[width=0.7\textwidth]{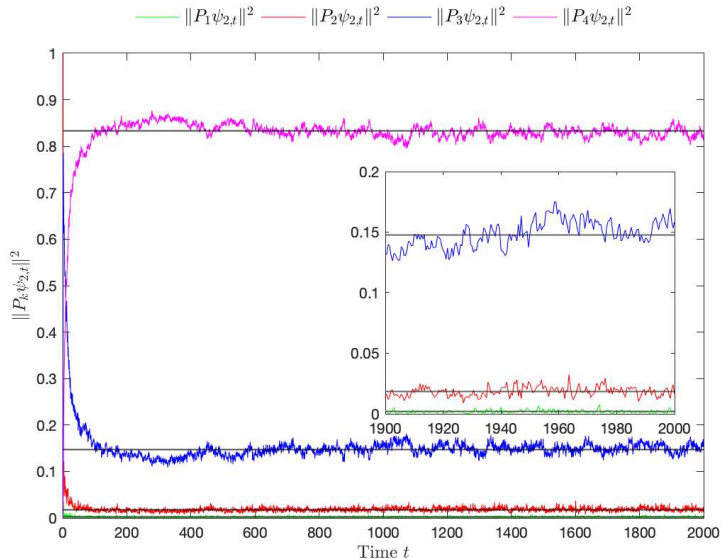}
 \caption{\label{fig2} The same simulation as in Figure~\ref{fig1}, only for longer times. The horizontal black lines indicate the values of the weights $M_{2\nu}$. The inset shows a part of the figure in magnification. Theorem~\ref{thm: GNT Pnu} states that the displayed behavior is typical of initial states in $\Hilbert_2$: up to fluctuations that are either small or rare, $\|P_\nu\psi_t\|^2$ is   close to $M_{2\nu}$.}
\end{figure}

For comparison, Reimann and Gemmer \cite{RG20} also concluded that $\scp{\psi_t}{A|\psi_t}$ is nearly constant, but for a different ensemble based on the condition $\scp{\psi_0}{A|\psi_0}\approx a$. We also provide a statement analogous to Theorem~\ref{thm: GNT Pnu} for $\scp{\psi_t}{B|\psi_t}$ with arbitrary observable $B$ instead of $P_\nu$ in Theorem~\ref{thm: GNT A} below.

\subsection{Example}

We illustrate Theorem~\ref{thm: GNT Pnu} within a simple random matrix model.
We partition the $D$-dimensional Hilbert space $\Hilbert:= \mathbb{C}^D =  \mathbb{C}^{d_1} \oplus \mathbb{C}^{d_2}\oplus \mathbb{C}^{d_3}\oplus \mathbb{C}^{d_4}=:\bigoplus_{\nu=1}^4\Hilbert_\nu$  into four macro spaces $\Hilbert_\nu$ of dimension $d_\nu$, i.e., $\Hilbert_1$ is spanned by the first $d_1$ canonical basis vectors, $\Hilbert_2$ by the next $d_2$ canonical basis vectors  and so on.
 The Hamiltonian $H$ is  a random $D\times D$-matrix $H$ that  has a band structure (i.e., mainly near-diagonal entries) and thus couples neighboring macro spaces more strongly than distant ones. 
  More precisely, we choose $H=(h_{ij})_{ij}$ to be a self-adjoint random matrix such that $h_{ii} \sim \mathcal{N}(0, \sigma_{ii}^2)$ and $h_{ij} \sim \mathcal{N}(0, \sigma_{ij}^2/2) + i \mathcal{N}(0, \sigma_{ij}^2/2)$ for $i\neq j$, where 
\begin{align}
    \sigma_{ij}^2 := \exp(-s|i-j|)
\end{align}
with some $s>0$ that controls the bandwidth. That is, the variances decrease exponentially in the distance from the diagonal.

In Figures~\ref{fig1} and \ref{fig2} the weights $\|P_\nu \psi_t\|^2$ are plotted for 
the values $s=0.02$, $d_\nu = 2 \times 10^{\nu-1}$, and 
 a random initial vector $\psi_0\in \Hilbert_2$. In Figure~\ref{fig1} the plot shows the initial phase where the system first passes through the 3rd macro state  before settling mostly in the ``equilibrium space'' $\Hilbert_4$. Note that the bandwidth is roughly $2s^{-1}= 400 \approx D^{0.77}\gg D^{0.5}$ and we thus expect to be in the regime of delocalized eigenfunctions, which is also confirmed by the numerical results.
 
Theorem \ref{thm: GNT Pnu} states that the long term behaviour depicted in Figure~\ref{fig2} is typical of initial states $\psi_0\in \Hilbert_2$: 
after some time the system equilibrates, the superposition weights $\|P_\nu\psi_{2,t}\|^2$ approach values $M_{2\nu}$ independent of the initial state, and stay close to them after the initial phase of equilibration. We also see that these values differ from the ones one would expect if normal typicality would hold: for example while in our simulation $d_4/D \approx 0.90$ one finds that $M_{24} \approx 0.82$.

 The average entropy as a function of time is plotted in Figure~\ref{fig3}. As expected, it increases   up to small fluctuations.

\begin{figure}[h]
\centering
 \includegraphics[width=0.7\textwidth]{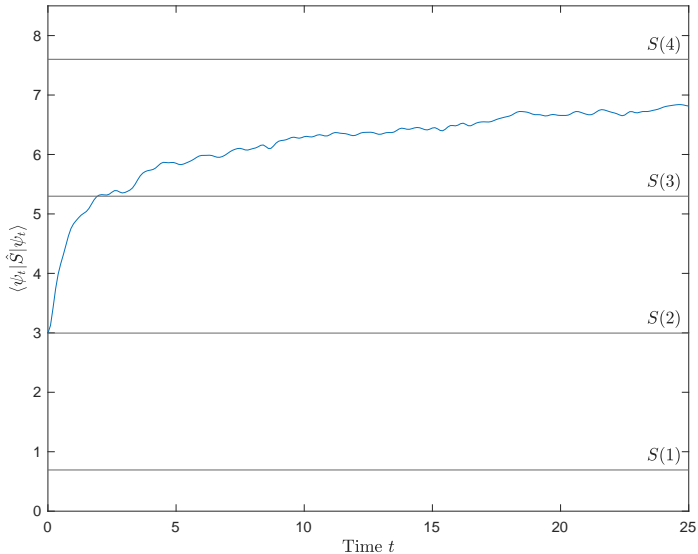}
 \caption{\label{fig3} The average entropy $\langle \psi_t|\hat{S}|\psi_t\rangle = \sum_\nu \|P_\nu\psi_t\|^2 S(\nu)$ as a function of time $t$ for $k_\mathrm{B}=1$ and the same simulation as in Figure~\ref{fig1} and Figure~\ref{fig2}. The tendency to increase can be regarded as a reflection of the second law of thermodynamics.
}
\end{figure}

\section{Proof of Theorem~\ref{thm:dyntyp}}
\label{sec:proofdyntyp}

The proof is very simple, based on an application of Chebyshev's, respectively Markov's, inequality to the following formulas for Hilbert space averages and Hilbert space variances \cite[App.~C]{GMM04}: For any Hilbert space $\Hilbert$ of dimension $d$, uniformly distributed $\psi\in\SSS(\Hilbert)$, and any operator $B$ on $\Hilbert$,
\begin{align}
\EEE \bigl[ \scp{\psi}{B|\psi} \bigr] &= \frac{1}{d} \tr B \label{Hilbertavg}\\
\Var \bigl[ \scp{\psi}{B|\psi} \bigr] &= \frac{1}{d(d+1)} \Bigl(\tr(B^\dagger B)-\frac{|\tr B|^2}{d} \Bigr)\,.\label{Hilbertvar}
\end{align}
(As usual, the variance of a complex random variable $Z$ is defined as Var $Z:= \mathbb{E}\bigl[|Z-\mathbb{E}(Z)|^2\bigr] = \mathbb{E}\bigl[|Z|^2\bigr]-|\mathbb{E}(Z)|^2$.)
Dropping the last term and replacing $d+1$ by $d$, we obtain the trivial upper bound
\be\label{Varbound1}
\Var \bigl[ \scp{\psi}{B|\psi} \bigr] \leq \frac{\tr(B^\dagger B)}{d^2} \,.
\ee
Now we insert $\Hilbert_\mu$ for $\Hilbert$ and $B=P_\mu \exp(iHt) P_\nu \exp(-iHt) P_\mu$; we write $\EEE_\mu$ and $\Var_\mu$ for expectation and variance over uniformly distributed $\psi_0\in\SSS(\Hilbert_\mu)$. We observe first that
\be 
\EEE_\mu \bigl[ \|P_\nu \psi_t\|^2 \bigr] 
= \frac{1}{d_\mu} \tr \bigl[P_\mu \exp(iHt) P_\nu \exp(-iHt) \bigr]
= w_{\mu\nu}(t)\,.
\ee
For the variance, since $|\tr(CD)|\leq \|C\| \tr(|D|)$ for any operators $C,D$ and $\|C\|$ the operator norm of $C$ \cite[Thm.~3.7.6]{simon}, we have that
\begin{align}
\tr(B^\dagger B) &= \tr \Bigl(P_\mu \exp(-iHt) P_\nu \exp(iHt) P_\mu \exp(iHt) P_\nu \exp(-iHt) P_\mu \Bigr)\\
&\leq \|P_\mu\| \, \|\exp(-iHt)\| \,\|P_\nu\| \, \|\exp(iHt)\| \cdots \|\exp(-iHt)\| \, \tr P_\mu\\
&= d_\mu\,.
\end{align}
We thus obtain that 
\be\label{Varbound4}
\Var_\mu \bigl[ \|P_\nu \psi_t\|^2 \bigr] \leq \frac{1}{d_\mu} \,.
\ee
The Chebyshev inequality then yields the first claim, \eqref{dyntyp1}.

For the second claim, Fubini's theorem allows us to interchange expectation and integral. Thus, 
\begin{align}
    \EEE_\mu \biggl[\int_0^T \bigl| \|P_\nu\psi_t\|^2 - w_{\mu\nu}(t) \bigr|^2\, dt\biggr]
    &=  \int_0^T \EEE_\mu \Bigl[ \bigl| \|P_\nu\psi_t\|^2 - w_{\mu\nu}(t) \bigr|^2\Bigr]\, dt\\
    &= \int_0^T \Var_\mu \bigl[ \|P_\nu\psi_t\|^2 \bigr]\, dt\\
    &\leq \frac{T}{d_\mu}
\end{align}
by \eqref{Varbound4}.
Markov's inequality then yields the second claim, \eqref{dyntyp2}.\hfill$\square$

\medskip

As a side remark, the arguments of the proof also yield the following upper bound on the Hilbert space variance over subspaces of dimension $d_\mu$ for arbitrary $B$:
\be
\Var_\mu \bigl[ \scp{\psi}{B|\psi} \bigr] \leq \frac{\tr(P_\mu B^\dagger P_\mu B P_\mu)}{d_\mu^2}
\leq \frac{\|B\| \, \tr(|B|)}{d_\mu^2}\,.
\ee

\section{More General Results}
\label{sec:proofstrategies}

\subsection{Dynamical Typicality}

Here is a variant of Theorem~\ref{thm:dyntyp} that allows for an arbitrary operator $B$ instead of $P_\nu$ and provides a tighter error bound:

\begin{thm}\label{thm:dyntypA}
Let $\mu,\nu$ be arbitrary macro states and let $B$ be any operator on $\Hilbert$. There is a function $w_{\mu B}: \mathbb{R}\to[0,1]$ such that for every $t\in\mathbb{R}$ and every $\varepsilon>0$, for $(1-\varepsilon)$-most $\psi_0\in\mathbb{S}(\Hilbert_\mu)$,
\be
\Bigl| \langle\psi_t|B|\psi_t\rangle-w_{\mu B}(t)\Bigr|\leq \min \Biggl\{ \frac{\|B\|}{\sqrt{\varepsilon d_\mu}},~ \sqrt{\frac{\|B\|\tr(|B|)}{\varepsilon d_\mu^2}},~ \sqrt{\frac{18\pi^3 \log(4/\varepsilon)}{d_\mu}}  \|B\| \Biggr\}.\label{dyntyp1A}
\ee
Moreover, for every $\mu$ and $B$, every $T>0$, and $(1-\varepsilon)$-most $\psi_0\in\mathbb{S}(\Hilbert_\mu)$,
\be
\frac{1}{T}\int_0^T \!  \bigl| \langle\psi_t|B|\psi_t\rangle-w_{\mu B}(t)\bigr|^2  dt \leq \frac{\|B\|^2}{\varepsilon d_\mu} \,.\label{dyntyp2A}
\ee
\end{thm}

In fact, the function $w_{\mu B}(t)$ is the average of $\langle\psi_t|B|\psi_t\rangle$ over $\psi_0\in\mathbb{S}(\Hilbert_\mu)$, which is
\be
w_{\mu B}(t) := \frac{1}{d_\mu} \tr\left[P_\mu \exp(iHt) B \exp(-iHt)\right].
\ee
The proof of Theorem~\ref{thm:dyntypA} (see Section~\ref{sec:proofdyntypA}) is largely analogous to that of Theorem~\ref{thm:dyntyp}. The bound involving $\sqrt{\log(1/\varepsilon)}$ instead of $1/\sqrt{\varepsilon}$ can be obtained by using L\'evy's lemma instead of the Chebyshev inequality. However, it turns out that for all other results in this paper, the bounds provided by Markov's and Chebyshev's inequality are better than those provided by L\'evy's lemma. That is because in many cases, L\'evy's lemma yields a bound that is better in $\varepsilon$ but worse in $d_\mu$, which in our situation is worse because $d_\mu$ is usually way larger than any relevant $1/\varepsilon$; see Section~\ref{sec:proofdyntypA} for more detail.

\subsection{Generalized Normal Typicality}

The next result, Theorem~\ref{thm: GNT A}, provides a somewhat more general version of Theorem~\ref{thm: GNT Pnu} that concerns arbitrary operators $B$ instead of $P_\nu$, as well as finite time intervals instead of $[0,\infty)$. To formulate it, we define
the number $d_E:=\#\mathcal{E}$ of distinct eigenvalues and
the maximal number of gaps in an energy  interval of length $\kappa>0$, 
\begin{align}
 G(\kappa) := \max_{E\in\mathbb{R} } \:\#\big\{&(e,e') \in \mathcal{E}\times \mathcal{E}\,:\, e\not=e'  \mbox{ and }e-e' \in [E,E+\kappa)\big\}\,.
\end{align}
It follows that $D_G = \lim_{\kappa \to 0^+} G(\kappa)$.
 
\begin{thm}\label{thm: GNT A}
Let $B$ be an operator on $\Hilbert$, let $\varepsilon,\delta,\kappa,T>0$, let $\mu$ be any macro state, and define  
\begin{align}
    M_{\mu B} &:= \frac{1}{d_\mu} \sum_{e\in\mathcal{E}} \tr\left(P_\mu \prj_e B \prj_e\right)\,.
\end{align}
Then $(1-\varepsilon)$-most $\psi_0\in\mathbb{S}(\Hilbert_\mu)$ are such that for $(1-\delta)$-most $t\in[0,T]$
\begin{align}\label{ineq: GNT A}
    \biggl|\langle&\psi_t|B|\psi_t\rangle - M_{\mu B} \biggr|  \leq\\ \nonumber&
      4 \,\sqrt{\frac{D_E \,G(\kappa) \|B\|}{\delta\varepsilon d_\mu} \left(1+\frac{8\log_2 d_E}{\kappa T}\right) \min\left\{\|B\|, \frac{\tr(|B|)}{d_\mu}\right\}}.
\end{align}
\end{thm}

Thus, as soon as 
 $d_\mu \gg D_E G(\kappa) \|B\|^2$
and $T$ is large enough, 
the right-hand side of \eqref{ineq: GNT A} is  small and the expectation $\langle\psi_t|B|\psi_t\rangle$ is close to a fixed value $M_{\mu B}$ for most times $t\in [0,T]$ and most initial states $\psi_0\in\mathbb{S}(\Hilbert_\mu)$. 
However, the times $T$ required to make the right-hand side of \eqref{ineq: GNT A} small are usually extremely large. For example, for a system of $N$ particles, $\Hilbert$
has dimension of the order $\exp(N)$; provided that no eigenvalue is hugely degenerate, there are of the order $\exp(N)$ energy eigenvalues. In order to obtain a small error, we need to keep $G(\kappa)$ small. For $\kappa\sim\exp(-N)\Delta E$, already the number of nearest-neighbor gaps with $e-e'\in[0,\kappa)$ will be of order $\exp(N)$, and will thus contribute of order $\exp(N)$ to $G(\kappa)$. So, we need $\kappa \ll \exp(-N)$ and therefore $T \gg \exp(N)$
to obtain a small error in \eqref{ineq: GNT A}.

For the proof of Theorems~\ref{thm: GNT Pnu} and \ref{thm: GNT A} we need, besides Hilbert space averages and variances, also Hilbert space covariances of two operators. The covariance of two complex random variables $X,Y$ is to be understood as
\begin{align}
\mathrm{Cov}[X,Y] &:= \EEE \bigl[ (X-\EEE X)^* (Y- \EEE Y) \bigr]\\
&= \EEE [X^*Y]- (\EEE X)^* \, \EEE Y\,.
\end{align}

\begin{lemma}[Hilbert Space Covariance]\label{lem: average}
For uniformly distributed $\psi\in\SSS(\Hilbert)$ with $\dim \Hilbert=d$ and any two operators $B,C$ on $\Hilbert$,
\begin{align}
\mathrm{Cov}\Bigl[ \scp{\psi}{B|\psi}, \scp{\psi}{C|\psi} \Bigr]
&= \frac{\tr(B^\dagger C)}{d(d+1)}- \frac{\tr(B^\dagger ) \tr(C)}{d^2(d+1)}\,.\label{covBC1}
\end{align}
\end{lemma}

Put differently,
\be\label{covBC2}
\EEE \bigl[ \scp{\psi}{B|\psi}^* \scp{\psi}{C|\psi} \bigr] = \frac{\tr(B^\dagger)\tr(C)+\tr(B^\dagger C)}{d(d+1)} \,.
\ee
By inserting $\Hilbert_\mu$ for $\Hilbert$, it follows that for uniformly distributed $\psi\in\SSS(\Hilbert_\mu)$ and any two operators $B,C$ on $\Hilbert$, 
\begin{align}
        \mathbb{E}_\mu\big[&\langle\psi|B^\dagger |\psi\rangle\langle \psi|C|\psi\rangle\big] = \label{covBC3}\\& \frac{1}{d_\mu(d_\mu+1)} \big(\tr(P_\mu B^\dagger ) \tr(P_\mu C) + \tr(P_\mu B^\dagger  P_\mu C)\big)\,.\nonumber
\end{align}

\section{Realistic Dimensions and Entropy}
\label{sec:realistic}

As indicated before, for a system of $N$ particles or more generally of $N$ degrees of freedom the dimension $D$ is of order $\exp(N)$.
We actually expect $D\approx \exp(s_\mathrm{eq} N/k_\mathrm{B})$, where $s_\mathrm{eq}$ is the entropy per particle in the thermal equilibrium state, and accordingly for all macro spaces $\Hilbert_\mu$,
\begin{align}\label{dimensions}
d_\mu =
\exp(s_\mu N/k_\mathrm{B}) \,.
\end{align}
The following corollary to Theorem~\ref{thm: GNT Pnu} shows that in this situation and assuming that no eigenvalues or gaps are macroscopically degenerate, fluctuations of the time-dependent superposition weights around their expected values are exponentially small in the number of particles with a rate controlled by the entropy per particle in the initial macro state. 

\begin{cor}\label{cor: GNT}
 Assume \eqref{dimensions}. 
 Then, for all   macro states $\mu,\nu_- ,\nu_+ $ with 
 \begin{align}
 s_{\nu_-} \leq s_\mu \leq s_{\nu_+}
 \end{align}
  it holds for $(1-\varepsilon)$-most $\psi_0\in\mathbb{S}(\Hilbert_\mu)$ for $(1-\delta)$-most of the time that
 \begin{align}\label{cor1}
     \biggl| \|P_{\nu_+}\psi_t\|^2 - M_{\mu\nu_+}  \biggr| &\leq \frac{4\sqrt{D_E D_G}}{\sqrt{\varepsilon\delta}} \,\exp\left(-\frac{s_\mu N}{2k_\mathrm{B}}\right) ,\\
     \biggl| \|P_{\nu_-}\psi_t\|^2 - M_{\mu\nu_-} \biggr| &\leq \frac{4\sqrt{D_E D_G}}{\sqrt{\varepsilon\delta}} \,\exp\left( -\frac{(s_\mu -\frac{s_{\nu_-}}{2})N}{k_\mathrm{B}}\right)\,.\label{cor2}
 \end{align}
\end{cor}

In particular, if $s_\mu,s_{\nu_{\pm}}$ are fixed and $N\to\infty$, the error bounds are exponentially small. Note also that the numerical experiment in Figure~\ref{fig2} is consistent with the idea that the fluctuations of the superposition weights in macro spaces $\nu_+$ of larger entropy than the initial state $\mu$ are controlled by the entropy $s_\mu$ of the initial macro state, while  the fluctuations of the superposition weights in macro spaces $\nu_-$ of smaller entropy than the initial state $\mu$ are   controlled by the entropy difference $s_\mu - s_{\nu_-}/2$ and thus even smaller.
 However, from the green line in Figure~\ref{fig2} (corresponding to $\|P_{1}\psi_t\|^2$) it is also apparent that the fluctuations of $\|P_{\nu}\psi_t\|^2$ might exceed the value of $M_{\mu\nu}$. Indeed, if we assume that the weights $M_{\mu\nu}$ scale like in the case of normal typicality, i.e.,
\begin{align}
   M_{\mu\nu} \approx \frac{d_\nu}{D} \approx 
   \exp\left(-\frac{s_\mathrm{eq} - s_\nu}{k_\mathrm{B}}N\right)\,,
\end{align}
then the relative error in \eqref{cor1} is only small if $s_{\nu_+} > s_\mathrm{eq} - s_\mu/2$,
and  the relative error in \eqref{cor2} is only small if $s_{\nu_-} > 2(s_\mathrm{eq} - s_\mu)$.

More generally, the question remains under which conditions one can prove that even for $M_{\mu\nu}$ close to 0, the relative error in \eqref{ineq: GNT Pnu} and thus the relative deviation of $\|P_\nu\psi_t\|^2$ from $M_{\mu\nu}$ will be small. In a separate work \cite{TTV22-mathe}, we study this question for specific distributions of the random matrix $H$.

\section{Outline of Proof of Theorem~\ref{thm: GNT A}}
\label{sec:outlinepfGNT}

Before we provide the technical details of the proof of Theorem \ref{thm: GNT A} in Section~\ref{sec:proofs}, we explain now the main strategy and the key ideas. 
The first step is to control the time variance
\begin{align}\label{Timevariance}
\left\langle\bigl|\langle\psi_t|B|\psi_t\rangle-M_{\psi_0 B}\bigr|^2\right\rangle_T := \frac{1}{T} \int_0^T \!\! \bigl|\langle\psi_t|B|\psi_t\rangle-M_{\psi_0 B}\bigr|^2 \,dt
\end{align}
of the quantity $\langle\psi_t|B|\psi_t\rangle$, where 
\begin{equation}
M_{\psi_0 B} =\overline{\langle\psi_t|B|\psi_t\rangle} := \lim_{T\to\infty} \frac{1}{T} \int_0^T\langle\psi_t|B|\psi_t\rangle\; dt
\end{equation}
is just the time-average of $\langle\psi_t|B|\psi_t\rangle$. The time variance \eqref{Timevariance} was the subject of several earlier investigations concerning thermalization in closed quantum systems. It is  usually controlled in terms of the effective dimension  \cite{Reimann08,Short11,SF12}
\be
d_{\textup{eff}} := \Bigl(\sum_e \langle\psi_0|\prj_e|\psi_0\rangle^2 \Bigr)^{-1}
\ee
of the initial state $\psi_0$, a measure for the number of distinct energies that contribute significantly to $\psi_0$.
In Section~\ref{sec:SF} we slightly improve the bound of \cite{SF12} (relevant when $d_\nu\ll d_\mu$) so that we can show that, after averaging the initial state over $\mathbb{S}(\Hilbert_\mu)$, one obtains that
  \begin{align}\label{Timevariance2}
    \mathbb{E}_\mu&\Bigl[ \left\langle\bigl|\langle\psi_t|B|\psi_t\rangle - M_{\psi_0 B}\bigr|^2\right\rangle_T \Bigr]\leq\\& \frac{2D_E G(\kappa)}{d_\mu+1} \left(1+\frac{8\log_2 d_E}{\kappa T}\right) \min\left\{\|B\|^2, \frac{\tr(B^\dagger  B)}{d_\mu}\right\}\,.\nonumber
    \end{align}

The second step is to show that $M_{\psi_0 B}$ is very close to $M_{\mu B}$ for most states $\psi_0\in\SSS(\Hilbert_\mu)$. To this end we observe that    
$ \mathbb{E}_\mu (M_{\psi_0 B})= M_{\mu B} $    and then bound the variance according to
 \begin{align}\label{muvariance}
   \mathbb{E}_\mu\Bigl[  (M_{\psi_0 B} - M_{\mu B})^2 \Bigr] &\leq \frac{\|B\|}{d_\mu+1} \min\left\{\|B\|, \frac{\tr(|B|)}{d_\mu} \right\}    \,.\end{align}
 A careful application of Markov's inequality then shows that \eqref{Timevariance2} and \eqref{muvariance} together imply \eqref{ineq: GNT A}.

\section{Remaining Proofs}
\label{sec:proofs}

\subsection{Proof of Theorem~\ref{thm:dyntypA}}
\label{sec:proofdyntypA}

The phenomenon of concentration of measure, i.e., that on a sphere in high dimension, ``nice'' functions are nearly constant, is often expressed by means of (e.g., \cite[Sec.~II.C]{SWGW22})

\begin{lemma}[L\'evy's Lemma]\label{lemma:Levy}
For any Hilbert space $\Hilbert$ with dimension $d$, any $f:\SSS(\Hilbert)\to\RRR$ with Lipschitz constant $\eta(f)$, and any $\varepsilon>0$,
\be\label{Levy}
|f(\psi)-\EEE f| \leq \sqrt{\frac{9\pi^3 \log(4/\varepsilon)}{2d}}\, \eta(f) 
\ee
for $(1-\varepsilon)$-most $\psi\in\SSS(\Hilbert)$.
\end{lemma}

Alternatively, Chebyshev's inequality yields that
\be\label{Chebybound}
|f(\psi)-\EEE f| \leq \sqrt{\frac{\Var(f)}{\varepsilon}}
\ee
for $(1-\varepsilon)$-most $\psi\in\SSS(\Hilbert)$. In the important special case $f\geq 0$, Markov's inequality yields that
\be\label{Markovbound}
f(\psi) \leq \frac{\EEE f}{\varepsilon}
\ee
for $(1-\varepsilon)$-most $\psi\in\SSS(\Hilbert)$, while L\'evy's lemma can be used in this situation to obtain that
\be\label{MarkovLevy}
f(\psi) \leq \EEE f + \sqrt{\frac{9\pi^3 \log(4/\varepsilon)}{2d}}\, \eta(f) \,.
\ee
Which bound is best depends on $\eta(f)$, $\Var(f)$, and $\EEE f$. For quadratic functions $f(\psi)=\scp{\psi}{B|\psi}$, $\eta(f)= 2\|B\|$ on $\SSS(\Hilbert)$, while expectation and variance are given by \eqref{Hilbertavg} and \eqref{Hilbertvar}; the first two bounds in \eqref{dyntyp1A} arise from the Chebyshev bound \eqref{Chebybound} with different ways of bounding the variance, and the third from L\'evy's lemma \eqref{Levy}.

As remarked already, the other results in this paper are not improved by using L\'evy's lemma instead of Markov's and Chebyshev's inequality. That is basically because the relevant functions $f\geq 0$ have means that are small like $1/$dimension but Lipschitz constants of order 1, so that \eqref{MarkovLevy} yields errors of order $1/\sqrt{\text{dimension}}$. Now it is of little interest to make $\varepsilon$ smaller than $10^{-200}$. (Borel once argued \cite[Chap.~6]{Bor62} that events with a probability of $10^{-200}$ or less can be expected to never occur in the history of the universe.)
On the other hand, the dimensions are large like $10^N$, so the advantage of \eqref{MarkovLevy} over \eqref{Markovbound} in $\varepsilon$ does not compensate for its disadvantage in the dimension.

\begin{proof}[Proof of Theorem \ref{thm:dyntypA}]
By \eqref{Hilbertavg} after inserting $\Hilbert_\mu$ for $\Hilbert$ and $P_\mu \exp(iHt) B \exp(-iHt) P_\mu$ for $B$,
\be
\EEE_\mu \scp{\psi_t}{B|\psi_t} = \frac{1}{d_\mu} \tr \bigl( P_\mu \exp(iHt) B \exp(-iHt) \bigr) = w_{\mu B}(t)\,.
\ee
L\'evy's lemma with $\eta=2\|B\|$ yields the third bound in \eqref{dyntyp1A}.

By \eqref{Varbound1} after inserting $\Hilbert_\mu$ for $\Hilbert$ and $P_\mu \exp(iHt) B \exp(-iHt) P_\mu$ for $B$,
\be\label{Varbound2}
\Var_\mu \scp{\psi_t}{B|\psi_t} \leq \frac{1}{d_\mu^2}\tr \Bigl( P_\mu \exp(-iHt) B^\dagger \exp(iHt) P_\mu \exp(iHt) B \exp(-iHt) P_\mu \Bigr) \,.
\ee
We give two upper bounds for the last expression. First, using $|\tr(CD)| \leq \|C\| \tr(|D|)$ and $\|B^\dagger\| = \|B\|$,
\begin{align}
\eqref{Varbound2} 
&\leq \frac{1}{d_\mu^2} \|P_\mu\| \|\exp(-iHt)\| \cdots \|\exp(-iHt)\| \tr P_\mu\\
&= \frac{1}{d_\mu^2} \|B\|^2 d_\mu = \frac{\|B\|^2}{d_\mu}\,.\label{Varbound3}
\end{align}
Second, by leaving $B$ rather than $P_\mu$ inside the trace,
\be
\eqref{Varbound2} \leq \frac{1}{d_\mu^2} \|B\| \tr (|B|)\,.
\ee
From these two bounds on the variance, \eqref{Chebybound} yields the first two bounds in \eqref{dyntyp1A}. For the second claim, \eqref{dyntyp2A}, of Theorem~\ref{thm:dyntypA}, the proof works as for Theorem~\ref{thm:dyntyp} with the bound \eqref{Varbound3} for $\Var_\mu \scp{\psi_t}{B|\psi_t}$.
\end{proof}

\subsection{Probability Current}
\label{sec:current}

In order to see that also the probability current $J_{\nu\nu'}(t)$ as defined in \eqref{Jdef} is deterministic, we verify that $\scp{\psi_t}{P_\nu H P_{\nu'}|\psi_t}$ is deterministic. This can be obtained in the same way as for Theorem~\ref{thm:dyntypA} by considering $B=P_\mu \exp(iHt) P_\nu H P_{\nu'} \exp(-iHt) P_\mu$ instead of $B=P_\mu \exp(iHt) P_\nu \exp(-iHt) P_\mu$ and noting that $\|B\| \leq \|H\| = \max\{|E-\Delta E|, |E|\}$. Physically, we expect $E$ to be comparable to the particle number $N$ and thus of order $\log D$, so $|J_{\nu\nu'}(t)-\EEE J_{\nu\nu'}(t)|$ is bounded by a constant times $\log D/\sqrt{\varepsilon d_\mu}$ (which would be small if we imagine $d_\mu \sim D^\alpha$ with $0<\alpha < 1$ and fixed $\varepsilon$) for $(1-\varepsilon)$-most $\psi_0 \in \SSS(\Hilbert_\mu)$. Likewise, $1/T$ times the $L^2$ norm over $[0,T]$ is bounded by a constant times $\log^2 D/\varepsilon d_\mu$ (which should be small).

\subsection{Hilbert Space Covariance}
\label{SM sec: Av Lemma}

For the proof of Lemma~\ref{lem: average}, we need the fourth moments of a random vector that is uniformly distributed over the unit sphere. So consider any Hilbert space $\Hilbert$ of dimension $d$ and a uniformly distributed $\psi\in\SSS(\Hilbert)$. Let $\left(\varphi_m\right)_{m}$ be an orthonormal basis of $\Hilbert$ and $a_m := \langle \varphi_m|\psi\rangle$. Then \cite{vonNeumann29}, \cite[App.~A.2 and C.1]{GMM04}

\begin{subequations}\label{expectations a_j}
\begin{align}
    (i)&\: \mathbb{E}(a_k^* a_l a_m^* a_n) = 0 \quad \mbox{if an index occurs only once},\\
    (ii)&\: \mathbb{E}\left(a_k^{*2} a_l^2\right) =0 \quad \mbox{for}\; k\neq l,\\
    (iii)&\: \mathbb{E}\left(|a_k|^4\right) = \frac{2}{d(d+1)},\\
    (iv)&\: \mathbb{E}\left(|a_k|^2 |a_l|^2\right) = \frac{1}{d(d+1)} \quad \mbox{for}\; k \neq l.
\end{align}
\end{subequations}

\begin{proof}[Proof of Lemma \ref{lem: average}]
Let $(\varphi_m)_m$ be an orthonormal basis of $\Hilbert$. Then we can write $\psi\in\mathbb{S}(\Hilbert)$ as 
 \begin{align}
     \psi = \sum_m a_m \varphi_m
 \end{align}
 with coefficients $a_m = \langle \varphi_m|\psi\rangle$. By \eqref{expectations a_j}, we get that
 \begin{align}
     \mathbb{E}\bigl[\langle\psi|B|\psi\rangle^* \langle\psi|C|\psi\rangle\bigr] &= \sum_{k,l,k',l'} \langle\varphi_k|B^\dagger|\varphi_l\rangle \langle\varphi_{k'}|C|\varphi_{l'}\rangle \mathbb{E}\left(a_{k}^* a_l a_{k'}^* a_{l'}\right)\\
     &= \frac{1}{d(d+1)} \sum_{k,l,k',l'} \langle\varphi_k|B^\dagger|\varphi_l\rangle \langle\varphi_{k'}|C|\varphi_{l'}\rangle \left(\delta_{kl} \delta_{k'l'}+\delta_{kl'}\delta_{k'l}\right)\\
     &= \frac{1}{d(d+1)} \left(\sum_{k,l} \langle\varphi_k|B^\dagger|\varphi_k\rangle \langle \varphi_l|C|\varphi_l\rangle + \langle\varphi_k|B^\dagger|\varphi_l\rangle \langle\varphi_l|C|\varphi_k\rangle\right)\\
     &= \frac{1}{d(d+1)} \bigl(\tr(B^\dagger) \tr(C) + \tr(B^\dagger C)\bigr).
 \end{align}
Thus,
\begin{align}
\mathrm{Cov}\bigl[ \scp{\psi}{B|\psi}, \scp{\psi}{C|\psi} \bigr]
&= \EEE \bigl[ \scp{\psi}{B|\psi}^* \scp{\psi}{C|\psi}\bigr] - \EEE\bigl[ \scp{\psi}{B|\psi}^*\bigr] \, \EEE \bigl[ \scp{\psi}{C|\psi} \bigr]\\
&= \frac{\tr(B^\dagger) \tr(C) + \tr(B^\dagger C)}{d(d+1)} -\frac{\tr(B^\dagger) \tr(C)}{d^2}\\
&= \frac{\tr(B^\dagger C)}{d(d+1)} - \frac{\tr(B^\dagger) \tr(C)}{d^2(d+1)}\,.
\end{align}
\end{proof}

\subsection{Computing and Estimating some Averages over $\boldsymbol{\mathbb{S}(\Hilbert_\mu)}$}

As a preparation for the proof of Theorem~\ref{thm: GNT A}, we derive in this section some upper bounds for relevant time and Hilbert space variances. We first note that it is well known that the limit in 
\begin{align}
    M_{\psi_0 B} = \overline{\langle\psi_t|B|\psi_t\rangle} := \lim_{T\to\infty} \frac{1}{T}\int_0^T \langle\psi_t|B|\psi_t\rangle\, dt\label{SM eq: time av}
\end{align}
exists for all $B$ and is given by
\be\label{timeavg}
M_{\psi_0 B} =  \scp{\psi_0}{\sum_{e\in\mathcal{E}}\prj_e B \prj_e|\psi_0}\,.
\ee
From \eqref{Hilbertavg}, applied to $\Hilbert_\mu$, we then obtain that
\be
\mathbb{E}_\mu M_{\psi_0 B} = \frac{1}{d_\mu}\sum_{e\in\mathcal{E}} \tr(P_\mu \prj_e B \prj_e)= M_{\mu B}\,.
\ee

\begin{prop}\label{SM prop upper bounds}
Let $\psi_0$ be uniformly distributed in $\mathbb{S}(\Hilbert_\mu)$, and let $B$ be any operator on $\Hilbert$. Then for every $\kappa, T>0$, 
\begin{align}
    \mathbb{E}_\mu\left(\left\langle\left|\langle\psi_t|B|\psi_t\rangle - \overline{\langle\psi_t|B|\psi_t\rangle}\right|^2\right\rangle_T\right)&\leq \frac{2D_E G(\kappa)}{d_\mu+1} \left(1+\frac{8\log_2 d_E}{\kappa T}\right) \min\left\{\|B\|^2, \frac{\tr(B^\dagger  B)}{d_\mu}\right\},\label{uppertimevar}\\
   \Var_\mu\overline{\langle\psi_t|B|\psi_t\rangle} &\leq \frac{\|B\|}{d_\mu+1} \min\left\{\|B\|, \frac{\tr(|B|)}{d_\mu} \right\}.
\end{align}
\end{prop}

\begin{proof}
  We start similarly to the proof of Theorem 1 in \cite{SF12} and compute 
\begin{align}
    \left\langle \left|\langle\psi_t|B|\psi_t\rangle - \overline{\langle\psi_t|B|\psi_t\rangle}\right|^2 \right\rangle_T 
    &= \left\langle \left| \sum_{e,e'} e^{i(e-e')t} \langle \psi_0 |\prj_e B \prj_{e'}|\psi_0\rangle  - \sum_e \langle \psi_0| \prj_e B \prj_e|\psi_0\rangle\right|^2 \right\rangle_T\\
    &= \left\langle \left| \sum_{e\neq e'} e^{i(e-e')t} \langle\psi_0| \prj_e B \prj_{e'}|\psi_0\rangle \right|^2 \right\rangle_T\\
    &= \sum_{\substack{e\neq e'\\ e''\neq e'''}} \left\langle e^{i(e-e'-e''+e''')t}\right\rangle_T \langle\psi_0| \prj_e B \prj_{e'}|\psi_0\rangle \langle\psi_0| \prj_{e'''} B^\dagger  \prj_{e''}|\psi_0\rangle. \label{eq: Vart proof}
\end{align}

By averaging over $\psi_0\in\mathbb{S}(\Hilbert_\mu)$, we obtain
\begin{align}
    &\mathbb{E}_\mu \left(\left\langle \left|\langle\psi_t|B|\psi_t\rangle - \overline{\langle\psi_t|B|\psi_t\rangle} \right|^2 \right\rangle_T \right) \nonumber\\
    &= \sum_{\substack{e\neq e'\\ e''\neq e'''}} \left\langle e^{i(e-e'-e''+e''')t} \right\rangle_T \mathbb{E}_\mu\Bigl[\langle\psi_0| \prj_e B \prj_{e'}|\psi_0\rangle \langle\psi_0| \prj_{e'''} B^\dagger  \prj_{e''}|\psi_0\rangle\Bigr]\\
    &= \frac{1}{d_\mu(d_\mu+1)} \sum_{\substack{e\neq e'\\ e''\neq e'''}} \left\langle e^{i(e-e'-e''+e''')t} \right\rangle_T \Bigl[\tr(P_\mu \prj_e B \prj_{e'}) \tr(P_\mu \prj_{e'''} B^\dagger  \prj_{e''}) \nonumber\\
    &\quad + \tr(P_\mu \prj_e B \prj_{e'} P_\mu \prj_{e'''} B^\dagger  \prj_{e''})\Bigr],
\end{align}
where we applied Lemma \ref{lem: average} in the form \eqref{covBC3} in the second equality.

Next we compute the ensemble variance of $\overline{\langle\psi_t|B|\psi_t\rangle}$: By \eqref{timeavg} and \eqref{Hilbertvar} for $\Hilbert_\mu$,
\begin{align}
    &\Var_\mu \overline{\langle\psi_t|B|\psi_t\rangle} \nonumber\\
    &= \frac{1}{d_\mu(d_\mu+1)} \sum_{e,e'} \tr(P_\mu \prj_e B \prj_e P_\mu \prj_{e'} B^\dagger  \prj_{e'}) - \frac{1}{d_\mu^2(d_\mu+1)} \left|\sum_e \tr(P_\mu \prj_e B \prj_e)\right|^2.
\end{align}

In the rest of the proof we use the computed expressions to prove the upper bounds for $\mathbb{E}_\mu \left(\langle \left|\langle\psi_t|B|\psi_t\rangle - M_{\psi_0 B}\right|^2 \rangle_T \right)$ and $ \Var_\mu \overline{\langle\psi_t|B|\psi_t\rangle}$. To this end, we define for $\alpha = (e,e') \in \mathcal{G} := \{(\bar{e},\bar{e}')\in\mathcal{E}\times\mathcal{E}, \bar{e}\neq\bar{e}'\}$ the vector $v_\alpha := \langle\psi_0| \prj_{e'} B^\dagger  \prj_e|\psi_0\rangle$. Moreover, we define the Hermitian matrix
\begin{align}
    R_{\alpha\beta} := \left\langle e^{i(G_\alpha-G_\beta)t} \right\rangle_T
\end{align}
with $G_\alpha:=e-e'$ for $\alpha=(e,e')$.
Then we obtain with \eqref{eq: Vart proof} that
\begin{align}
    \left\langle \left|\langle\psi_t|B|\psi_t\rangle - \overline{\langle\psi_t|B|\psi_t\rangle}\right|^2 \right\rangle_T &= \sum_{\alpha,\beta} v_\alpha^* R_{\alpha\beta} v_\beta\\
    &\leq \|R\| \sum_{\alpha} |v_\alpha|^2\\
    &= \|R\| \sum_{e\neq e'} \bigl|\langle\psi_0|\prj_e B \prj_{e'}|\psi_0\rangle\bigr|^2
\end{align}
and thus
\begin{align}
    &\mathbb{E}_\mu\left(\left\langle \left|\langle\psi_t|B|\psi_t\rangle - \overline{\langle\psi_t|B|\psi_t\rangle} \right|^2 \right\rangle_T\right)\nonumber\\
    &\leq \|R\| \sum_{e,e'} \mathbb{E}_\mu\bigl[\langle\psi_0|\prj_e B \prj_{e'}|\psi_0\rangle \langle\psi_0|\prj_{e'} B^\dagger  \prj_e|\psi_0\rangle\bigr]\\
    &= \frac{\|R\|}{d_\mu(d_\mu+1)} \sum_{e,e'} \left[\left|\tr(P_\mu \prj_e B \prj_{e'})\right|^2 + \tr(P_\mu \prj_e B \prj_{e'} P_\mu \prj_{e'} B^\dagger  \prj_e)\right]
\end{align}
by \eqref{covBC3}.
Short and Farrelly \cite{SF12} showed for arbitrary $\kappa>0$ and $T>0$ that
\begin{align}
    \|R\| \leq G(\kappa) \left(1+\frac{8\log_2 d_E}{\kappa T}\right).
\end{align}
Moreover, we estimate
\begin{align}
    \sum_{e,e'} |\tr(P_\mu \prj_e B \prj_{e'})|^2 &= \sum_{e,e'} |\tr(\prj_{e'} P_\mu \prj_e \prj_e B \prj_{e'})|^2\\
    &\leq \sum_{e,e'} \underbrace{\tr(\prj_{e'} P_\mu \prj_e P_\mu)}_{\leq \tr(\prj_{e'})\leq D_E} \tr(\prj_{e'} B^\dagger  \prj_e B)\\
    &\leq D_E \tr(B^\dagger  B),
\end{align}
where we used the Cauchy-Schwarz inequality for operators $A,B$ with scalar product $\tr(A^\dagger  B)$ and that $|\tr(CD)|\leq \|C\| \tr(|D|)$. Similarly we find that
\begin{align}
    \sum_{e,e'} |\tr(P_\mu \prj_e A \prj_{e'})|^2 &\leq \sum_{e,e'} \tr(\prj_{e'}P_\mu \prj_e P_\mu) \underbrace{\tr(\prj_{e'} B^\dagger  \prj_e B)}_{\leq \tr(\prj_{e'})\|B\|^2}\\
    &\leq D_E \|B\|^2 d_\mu.
\end{align}
 This shows that
\begin{align}
    \sum_{e,e'} |\tr(P_\mu \prj_e B \prj_{e'})|^2 &\leq D_E \min\{\|B\|^2 d_\mu, \tr(B^\dagger B)\}.
\end{align}
Next we compute
\begin{align}
    \sum_{e,e'} \tr(P_\mu \prj_e B \prj_{e'} P_\mu \prj_{e'} B^\dagger  \prj_e) &= \sum_e \tr\left(\prj_e P_\mu \prj_e B \left(\sum_{e'} \prj_{e'}P_\mu \prj_{e'}\right) B^\dagger \right)\\
    &\leq \sum_{e} \tr(\prj_e P_\mu \prj_e) \left\|B \left(\sum_{e'} \prj_{e'} P_\mu \prj_{e'}\right) B^\dagger \right\|\\
    &\leq \|B\|^2 \sum_{e} \tr(\prj_e P_\mu)\\
    &= \|B\|^2 d_\mu,
\end{align}
where we used in the third line that $\|\sum_{e'} \prj_{e'}P_\mu \prj_{e'}\| \leq 1$, which follows immediately from
\begin{align}
    \left\|\sum_{e'} \prj_{e'}P_\mu \prj_{e'} \psi_0 \right\|^2 = \sum_{e'} \|\prj_{e'}P_\mu \prj_{e'}\psi_0\|^2 &\leq \sum_{e'} \|\prj_{e'}\psi_0\|^2 = \|\psi_0\|^2. 
\end{align}
Similarly we estimate
\begin{align}
  \sum_{e,e'} \tr(P_\mu \prj_e B \prj_{e'} P_\mu \prj_{e'} B^\dagger  \prj_e) &= \sum_{e'} \tr\left(\left(\sum_e \prj_e P_\mu \prj_e\right)B \prj_{e'} P_\mu \prj_{e'} B^\dagger  \right)\\
  &\leq \sum_{e'} \tr(B \prj_{e'} P_\mu \prj_{e'} B^\dagger )\\
  &= \sum_{e'} \tr(\prj_{e'}B^\dagger  B \prj_{e'} \prj_{e'} P_\mu \prj_{e'})\\
  &\leq \sum_{e'} \tr(\prj_{e'} B^\dagger  B)\\
  &= \tr(B^\dagger B).
\end{align}
The previous two estimates show that
\begin{align}
      \sum_{e,e'} \tr(P_\mu \prj_e B \prj_{e'} P_\mu \prj_{e'} B^\dagger  \prj_e) \leq \min\{\|B\|^2 d_\mu, \tr(B^\dagger B)\}.
\end{align}
Putting everything together, we arrive at the upper bound
\begin{align}
    \mathbb{E}_\mu\left(\left\langle \left|\langle\psi_t|B|\psi_t\rangle - \overline{\langle\psi_t|B|\psi_t\rangle} \right|^2 \right\rangle_T\right) \leq \frac{2D_E G(\kappa)}{d_\mu+1} \left(1+\frac{8\log_2 d_E}{\kappa T}\right) \min\left\{\|B\|^2, \frac{\tr(B^\dagger  B)}{d_\mu}\right\}.
\end{align}
Finally we turn to the upper bound for $\Var_\mu\overline{\langle\psi_t|B|\psi_t\rangle}$. To this end, we estimate
\begin{align}
    \sum_{e,e'}\tr(P_\mu \prj_e B \prj_e P_\mu \prj_{e'} B^\dagger  \prj_{e'}) &= \tr\left(P_\mu \left(\sum_e \prj_e B \prj_e\right)P_\mu \left(\sum_{e'} \prj_{e'} B^\dagger  \prj_{e'}\right)\right)\\
    &\leq \tr(P_\mu) \left\|\left(\sum_e \prj_e B \prj_e\right)P_\mu \left(\sum_{e'} \prj_{e'} B^\dagger  \prj_{e'}\right) \right\|\\
    &\leq d_\mu \|B\|^2
\end{align}
and 
\begin{align}
     \sum_{e,e'}\tr(P_\mu \prj_e B \prj_e P_\mu \prj_{e'} B^\dagger  \prj_{e'}) &=  \tr\left(B \left(\sum_e \prj_e P_\mu \left(\sum_{e'} \prj_{e'} B^\dagger  \prj_{e'}\right)P_\mu \prj_e\right)\right)\\
     &\leq \tr(|B|) \left\|\sum_e \prj_e P_\mu \left(\sum_{e'} \prj_{e'} B^\dagger  \prj_{e'}\right)P_\mu \prj_e\right\|\\
     &\leq \tr(|B|) \left\|P_\mu\left(\sum_{e'} \prj_{e'}B^\dagger  \prj_{e'}\right)P_\mu \right\|\\
     &\leq \tr(|B|)\|B\|.
\end{align}
This shows that
\begin{align}
    \sum_{e,e'}\tr(P_\mu \prj_e B \prj_e P_\mu \prj_{e'} B^\dagger  \prj_{e'}) &\leq \|B\| \min\left\{d_\mu \|B\|, \tr(|B|)\right\}
\end{align}
and thus
\begin{align}
    \Var_\mu\overline{\langle\psi_t|B|\psi_t\rangle} &\leq \frac{1}{d_\mu(d_\mu+1)} \sum_{e,e'}\tr(P_\mu \prj_e B \prj_e P_\mu \prj_{e'} B^\dagger  \prj_{e'})\\
    &\leq \frac{\|B\|}{d_\mu+1} \min\left\{\|B\|, \frac{\tr(|B|)}{d_\mu} \right\}.
\end{align}
\end{proof}

\subsection{Proof of Theorems \ref{thm: GNT Pnu} and \ref{thm: GNT A}}

Theorem~\ref{thm: GNT Pnu} follows immediately from Theorem~\ref{thm: GNT A} by setting
$B=P_\nu$, choosing $\kappa$ small enough such that $G(\kappa)=D_G$, and then taking the limit $T\to\infty$. 

\begin{proof}[Proof of Theorem~\ref{thm: GNT A}]
Markov's inequality implies
\begin{align}
     &\mathbb{P}_\mu\left(\left\langle \left|\langle\psi_t|B|\psi_t\rangle - M_{\psi_0 B} \right|^2 \right\rangle_T \geq \frac{4D_E G(\kappa) \|B\|}{\varepsilon d_\mu} \left(1+\frac{8\log_2 d_E}{\kappa T}\right)\min\left\{\|B\|, \frac{\tr(|B|)}{d_\mu} \right\}\right)\nonumber\\ 
     &\leq \frac{\mathbb{E}_\mu\left(\left\langle \left| \langle\psi_t|B|\psi_t\rangle - M_{\psi_0 B}\right|^2 \right\rangle_T\right)}{4 D_E G(\kappa) \|B\| \left(1+\frac{8\log_2 d_E}{\kappa T}\right) \min\left\{\|B\|, \frac{\tr(|B|)}{d_\mu} \right\}} \varepsilon d_\mu\\
     &\leq \frac{\min\left\{\right\|B\|^2, \frac{\tr(B^\dagger B)}{d_\mu}\}}{2\|B\| \min\left\{\|B\|, \frac{\tr(|B|)}{d_\mu}\right\}}\varepsilon\\
     &\leq \frac{\varepsilon}{2},
\end{align}
where we used the bounds from Proposition \ref{SM prop upper bounds} and that $\tr(B^\dagger B) \leq \|B\| \tr(|B|)$. This means that for $(1-\frac{\varepsilon}{2})$-most $\psi_0\in\mathbb{S}(\Hilbert_\mu)$,
\begin{align}
    \left\langle \left|\langle\psi_t|B|\psi_t\rangle - M_{\psi_0 B} \right|^2 \right\rangle_T < \frac{4D_E G(\kappa) \|B\|}{\varepsilon d_\mu} \left(1+\frac{8\log_2 d_E}{\kappa T}\right) \min\left\{\|B\|, \frac{\tr(|B|)}{d_\mu} \right\}.
\end{align}

Again with the help of Markov's inequality we obtain that, with $\lambda$ the Lebesgue measure on $\mathbb{R}$,
\begin{align}
 &\frac{\lambda\Bigl\{t\in[0,T]:\left|\langle\psi_t|B|\psi_t\rangle - M_{\psi_0 B} \right|^2 \geq \frac{4D_E G(\kappa) \|B\|}{\delta\varepsilon d_\mu} \left(1+\frac{8\log_2 d_E}{\kappa T}\right) \min\left\{\|B\|, \frac{\tr(|B|)}{d_\mu}\right\}\Bigr\}}{T}\\ &\hspace{5cm}\leq \frac{\delta\varepsilon d_\mu\left\langle\left|\langle\psi_t|B|\psi_t\rangle - M_{\psi_0 B} \right|^2\right\rangle_T}{4D_E G(\kappa) \|B\| \left(1+\frac{8\log_2 d_E}{\kappa T}\right) \min\left\{\|B\|, \frac{\tr(|B|)}{d_\mu} \right\}}\\
&\hspace{5cm}\leq \delta.
\end{align}
This shows that
for $(1-\frac{\varepsilon}{2})$-most $\psi_0\in\mathbb{S}(\Hilbert_\mu)$ we have for $(1-\delta)$-most $t\in[0,T]$ that
\begin{align}
    \left|\langle\psi_t|B|\psi_t\rangle - M_{\psi_0 B}\right| \leq 
    2\left(\frac{D_E G(\kappa) \|B\|}{\delta\varepsilon d_\mu}\left(1+\frac{8\log_2 d_E}{\kappa T}\right)\min\left\{\|B\|, \frac{\tr(|B|)}{d_\mu}\right\}\right)^{1/2}.
\end{align}
Next we prove in a similar way an upper bound for $|M_{\psi_0 B}-M_{\mu B}|$, keeping in mind that $M_{\mu B} = \mathbb{E}_\mu M_{\psi_0 B}$. An application of Chebyshev's inequality and Proposition \ref{SM prop upper bounds} shows that
\begin{align}
    \mathbb{P}_\mu\left(|M_{\psi_0 B}-M_{\mu B}|\geq \sqrt{\frac{2 \|B\| \min\left\{\|B\|,\frac{\tr(|B|)}{d_\mu}\right\}}{d_\mu\varepsilon}}\right)&\leq \frac{\Var_\mu\overline{\langle\psi_t|B|\psi_t\rangle}}{2\|B\|\min\left\{\|B\|, \frac{\tr(|B|)}{d_\mu}\right\}} d_\mu\varepsilon\\
    &\leq \frac{\varepsilon}{2}.
\end{align}
This implies for $(1-\frac{\varepsilon}{2})$-most $\psi_0\in\mathbb{S}(\Hilbert_\mu)$ that
\begin{align}
    |M_{\psi_0 B}-M_{\mu B}| \leq
    \sqrt{2}\left(\frac{\|B\|}{\varepsilon d_\mu}\min\left\{\|B\|,\frac{\tr(|B|)}{d_\mu}\right\}\right)^{1/2}.
\end{align}
With the triangle inequality we finally obtain the stated upper bound for $|\langle\psi_t|B|\psi_t\rangle - M_{\mu B}|$.
\end{proof}

\subsection{Proof of Corollary \ref{cor: GNT}}

From Theorem \ref{thm: GNT A} we obtain immediately that for $(1-\varepsilon)$-most $\psi_0\in\mathbb{S}(\Hilbert_\mu)$ for $(1-\delta)$-most of the time
\begin{align}
    \biggl|\|P_{\nu_+}\psi_t\|^2 - M_{\mu\nu_+}\biggr| &\leq 4\sqrt{\frac{D_E D_G}{\delta\varepsilon}} \exp\left(-\frac{s_\mu N}{2k_B}\right) \min\left\{1,\exp\left(\frac{(s_{\nu_+}-s_\mu)N}{2k_B}\right)\right\}\\
    &= 4 \frac{\sqrt{D_E D_G}}{\sqrt{\varepsilon\delta}} \exp\left(-\frac{s_\mu N}{2k_B}\right).
\end{align}
Similarly, we find for $(1-\varepsilon)$-most $\psi_0\in\mathbb{S}(\Hilbert_\mu)$ for $(1-\delta)$-most of the time that
\begin{align}
    \biggl|\|P_{\nu_-}\psi_t\|^2 - M_{\mu\nu_-}\biggr| &\leq 4\sqrt{\frac{D_E D_G}{\varepsilon\delta}} \exp\left(-\frac{s_\mu N}{2k_B}\right) \min\left\{1,\exp\left(\frac{(s_{\nu_-}-s_\mu)N}{2k_B}\right)\right\}\\
    &= 4\frac{\sqrt{D_E D_G}}{\sqrt{\varepsilon\delta}} \exp\left(-\frac{(s_\mu-\frac{s_{\nu_-}}{2})N}{k_B}\right).
\end{align}
This finishes the proof.$\hfill\square$

\subsection{Alternative Estimate in Terms of Effective Dimension}
\label{sec:SF}

In Proposition~\ref{SM prop upper bounds}, we have provided two upper bounds \eqref{uppertimevar} for
\[
\mathbb{E}_\mu\left(\left\langle \left|\langle\psi_t|B|\psi_t\rangle - \overline{\langle\psi_t|B|\psi_t\rangle} \right|^2 \right\rangle_T\right)\,.
\]
There is an alternative way of obtaining one of the two  bounds
in \eqref{uppertimevar} using a result of Short and Farrelly \cite{SF12} based on the concept of effective dimension. We briefly explain this alternative derivation and then comment on why we also need the other bound in~\eqref{uppertimevar}.

In \cite{SF12} the authors show that \begin{align}
    \left\langle \left|\langle\psi_t|B|\psi_t\rangle - \overline{\langle\psi_t|B|\psi_t\rangle} \right|^2 \right\rangle_T & \leq \frac{G(\kappa) \|B\|^2}{d_{\textup{eff}}} \left(1+\frac{8\log_2 d_E}{\kappa T}\right),\label{ineq: SF}
\end{align}
where the effective dimension $d_{\textup{eff}} = d_{\textup{eff}}(\psi_0)$ of a state $\psi_0$ is 
\begin{align}
    d_{\textup{eff}} = \left(\sum_e \langle\psi_0|\Pi_e|\psi_0\rangle^2\right)^{-1}.
\end{align}
Taking an average over $\psi_0\in\mathbb{S}(\Hilbert_\mu)$
yields the bound
\begin{align}
    \mathbb{E}_\mu\left(\left\langle \left|\langle\psi_t|B|\psi_t\rangle - \overline{\langle\psi_t|B|\psi_t\rangle} \right|^2 \right\rangle_T\right) &\leq \frac{2D_E G(\kappa) \|B\|^2}{d_\mu+1}\left(1+\frac{8\log_2 d_E}{\kappa T}\right).\label{ineq: Emu SFbound}
\end{align}
To see this, note that the only quantity on the right-hand side of \eqref{ineq: SF} that depends on $\psi_0$  is the effective dimension $d_{\textup{eff}}$; therefore, it suffices to estimate $\mathbb{E}_\mu d_{\textup{eff}}^{-1}$. With the help of \eqref{covBC3} and the usual arguments we find
\begin{align}
    \mathbb{E}_\mu d_{\textup{eff}}^{-1} &= \sum_e \mathbb{E}_\mu\left(\langle\psi_0|\Pi_e|\psi_0\rangle\langle\psi_0|\Pi_e|\psi_0\rangle\right)\\
    &= \frac{1}{d_\mu(d_\mu+1)}\left(\tr(P_\mu\Pi_e)^2+\tr(P_\mu\Pi_e P_\mu \Pi_e)\right)\\
    &\leq \frac{1}{d_\mu(d_\mu+1)}\sum_e \left(\underbrace{\tr(\Pi_e)}_{\leq D_E}\tr(P_\mu\Pi_e)+ \tr(P_\mu\Pi_eP_\mu)\right)\\
    &\leq \frac{2D_E}{d_\mu(d_\mu+1)}\sum_e \tr(P_\mu \Pi_e)\\
    &= \frac{2D_E}{d_\mu+1}\,,
\end{align}
and \eqref{ineq: Emu SFbound} immediately follows.

The second estimate in Proposition \ref{SM prop upper bounds} is sharper than \eqref{ineq: Emu SFbound}  if and only if $\tr(B^\dagger B)/d_\mu < \|B\|^2$, i.e., roughly speaking, if only few (compared to $d_\mu$) eigenvalues of $B^\dagger B$ are close to the largest eigenvalue and most are much smaller.
This becomes relevant, for example, when estimating the transitions from $\Hilbert_\mu$ into a lower entropy macro space $\Hilbert_\nu$, cf.\ \eqref{cor2}. Then $B=P_\nu$ and 
\[
\tr(B^\dagger B)/d_\mu = d_\nu/d_\mu \ll 1 =  \|P_\nu\|^2\,.
\]

\section{Conclusions}
\label{sec:conclusions}

Our results concern the behavior of typical pure states $\psi_0$ from a high-dimensional subspace $\Hilbert_\mu$ of Hilbert space under the unitary time evolution. We find that for any operator $B$, due to the large dimension, the curve $t\mapsto \scp{\psi_t}{B|\psi_t}$ is nearly deterministic (a fact that can also be obtained from \cite{BRGSR18,RG20}), and that in the long run $t\to\infty$ it is nearly constant.  In von Neumann's framework of an orthogonal decomposition $\Hilbert=\oplus_\nu \Hilbert_\nu$ into macro spaces, this means that the time-dependent distribution over the macro states given by the superposition weights $\|P_\nu \psi_t\|^2$ is nearly deterministic and in the long run nearly constant, i.e., it reaches \emph{normal equilibrium}, a situation analogous (but not identical) to thermal equilibrium. Through our theorems, we have provided explicit error bounds.

Von Neumann's \cite{vonNeumann29} prior result in the same direction was based on unrealistic assumptions, saying essentially that $H$ is unrelated to $\Hilbert_\nu$. Our result has the advantage of being applicable regardless of relations between $H$ and $\Hilbert_\nu$. The question of whether the deviation from the mean is small compared to the mean even when the mean is small itself, will be analyzed further elsewhere \cite{TTV22-mathe}. 
\\
\\
\\
\noindent \textbf{\large Acknowledgments}\\
\\
We thank both referees for valuable feedback and for pointing out to us reference~\cite{BRGSR18}.
C.V.\ gratefully acknowledges financial support by the German Academic Scholarship Foundation.
\\
\\
\noindent\textbf{\large Data Availability Statement}\\
\\
The Matlab code used to generate the datasets of the provided examples is available from the corresponding author on request.\\
\\
\noindent\textbf{\large Conflict of Interest Statement}\\
\\
The authors have no conflicts of interest.

\bibliographystyle{plainurl}

\end{document}